\documentclass[a4paper,12pt]{article}
\usepackage{cmap}
\usepackage[fleqn]{amsmath}
\usepackage{graphicx}
\usepackage{amsthm, amstext}
\usepackage{amssymb, array, amsfonts}
\usepackage[cp1251,cp866]{inputenc}
\usepackage{type1ec}
\usepackage[english]{babel}
\usepackage{epsfig}
\inputencoding{cp1251}
\usepackage{color}

\newtheorem{lemma}{Lemma}
\newtheorem{theorem}{Theorem}
\newtheorem{remark}{Remark}
\newtheorem{proposition}{Proposition}

\title{Connection problem for the tau-function of the Sine-Gordon reduction of Painlev\'{e}-III equation via the Riemann-Hilbert approach. } 
\date{}
\begin{document}
\maketitle
\centerline{ A.~Its}
\centerline{{\it Department of Mathematical Sciences,}}
\centerline{{\it Indiana University -- Purdue University  Indianapolis}}
\centerline{{\it Indianapolis, IN 46202-3216, USA}}
\vskip .2in
\centerline{ A.~Prokhorov}
\centerline{{\it Department of Mathematical Sciences,}}
\centerline{{\it Indiana University -- Purdue University  Indianapolis}}
\centerline{{\it Indianapolis, IN 46202-3216, USA}}

\bigskip\bigskip\bigskip

\noindent{\bf Abstract.}
We evaluate explicitly, in terms of the Cauchy data, the constant pre-factor in the large $x$  asymptotics
of the Painlev\'e III tau-function. Our result proves the conjectural formula for this pre-factor obtained
recently by O. Lisovyy,  Y. Tykhyy, and the first co-author with the help of the recently discovered
connection of the Painlev\'e tau-functions with the Virasoro conformal blocks. Our approach  does not use this connection,
and it is based  on the Riemann-Hilbert method.

\bigskip\bigskip\bigskip

\section{Introduction}
We consider the  particular  case of the third Painlev\'e equation,
which is  a radial-symmetric reduction of the elliptic sine-Gordon equation
\begin{gather}\label{sineP3}
u_{xx} + \frac{u_x}{x} + \sin u = 0.
\end{gather}
Starting from the pioneering works \cite{barouch},  \cite{MTW}  on the Ising model, this equation
has been playing an increasingly important role, as a ``nonlinear Bessel function'',
in  a growing  number of physical applications (see e.g., \cite{FIKN} and
references therein). Apparently, the first appearance of equation (\ref{sineP3}) in the physical applications  should be
credited to work of J. M. Myers \cite{My}.

Equation (\ref{sineP3}) can be written as a  (non-autonomous) Hamiltonian system,
\begin{gather}\label{hform}
\frac{du}{dx} = \frac{\partial{\mathcal{H}}}{\partial v}, \quad \quad \frac{dv}{dx} = -\frac{\partial{\mathcal{H}}}{\partial u},\nonumber
\end{gather}
on the phase space
${\mathbb{R}}^2 = \{(u,v)\}$ equipped with the canonical symplectic structure,
\begin{equation}\label{symp}
\Omega = dv\wedge du.
\end{equation}
The Hamiltonian ${\mathcal{H}}$ is given  by the formula
\begin{equation}\label{ham}
{\mathcal{H}}= \frac{v^2}{2x} - x\cos u.\nonumber
\end{equation}
We are concerned with the global asymptotic analysis of the $\tau$-{\it function} corresponding
to   the Painlev\'e III equation (\ref{sineP3}) which is defined according to the equation (see \cite{JM}, \cite{Ok}),
\begin{equation}\label{taudef}
\frac{d\ln\tau }{dx} = -\frac{1}{4}{\mathcal{H}}.
\end{equation}
In fact, it is  this tau-function (evaluated for a special family of  solutions of equation (\ref{sineP3}))
that played a key role in the above mentioned Barouch-McCoy-Tracy-Wu  theory,  and, since then,
it has appeared in many problems of  statistical mechanics and quantum field theory.
Let us now remind some of the basic   known facts about the asymptotics of the solutions  of equation (\ref{sineP3}).
We refer the reader to monograph \cite{FIKN} for more details and for the history of the question.

Equation (\ref{sineP3}) possesses  a  two-parameter family of solutions  characterized by the
following behavior at $x=0$,
\begin{equation}\label{atzero}
u(x) = \alpha\ln x + \beta +O\left(x^{2-|\Im \alpha|}\right), \quad x \to 0, 
\end{equation}
where the complex numbers $\alpha\in\mathbb{C}$, $|\Im{\alpha}| < 2$ and   $\beta\in\mathbb{C}$  can be taken as parameters - the Cauchy data,  of the solution $u(x)\equiv u(x|\alpha, \beta)$. The behavior of the solution $u(x|\alpha, \beta)$ as $x\to +\infty$, is known. For an open set  in the space of
parameters $\alpha$, $\beta$, which we will describe later,  the large $x$ behavior of $u(x|\alpha, \beta)$  is oscillatory, and it is given by the formulae,
$$
u(x) = b_+e^{ix}x^{i\nu - 1/2} \left( 1 + O\left(\frac{1}{x}\right)\right)
+ b_-e^{-ix}x^{-i\nu - 1/2} \left( 1 + O\left(\frac{1}{x}\right)\right) + 
$$
\begin{equation}\label{atinfty}
+ O\left(x^{3|\Im\nu| - 3/2}\right)\, (\mbox{mod}\,2\pi), \quad x \to \infty,
\end{equation}
where 
\begin{equation}\label{nu}
\nu = -\frac{1}{4}b_+b_-, \quad |\Im \nu| < 1/2.
\end{equation}
The asymptotic parameters at infinity - the complex amplitudes $b_{\pm}$, can be, in
fact, expressed in terms of the Cauchy data $\alpha$, $\beta$,  and the  condition $|\Im \nu| < 1/2$  imposes additional restriction on them. 
The corresponding {\it connection formulae} were obtained in 1985 by V. Yu. Novokshenov \cite{Nov} (see also: \cite{IN} and \cite{Kit}), 
and they are given by the equations,
\begin{equation}\label{connect}
 e^{\pi\nu} = \frac{\sin2\pi\eta}{\sin2\pi\sigma},\quad
 b_{\pm} = -e^{\frac{\pi\nu}{2} \mp \frac{i\pi}{4}}2^{1\pm2i\nu}\frac{1}{\sqrt{2\pi}}\Gamma(1\mp i\nu)
 \frac{\sin 2\pi(\sigma\mp \eta)}{\sin2\pi\eta},
 \end{equation}
where
 \begin{equation}\label{sigma}
 \sigma:= \frac{1}{4} +\frac{i}{8}\alpha,\quad
 \eta := \frac{1}{4}+\frac{1}{4\pi}\Bigl(\beta  + \alpha\ln 8\Bigr) +\frac{i}{2\pi}
\ln\frac{\Gamma\left(\frac{1}{2}-\frac{i\alpha}{4}\right)}{\Gamma\left(\frac{1}{2}+\frac{i\alpha}{4}\right)},
\end{equation}
and $\Gamma(z)$ is Euler's Gamma-function. The open set in the space of the Cauchy data $\alpha$, $\beta$ where the both
asymptotics,  (\ref{atzero}) and (\ref{atinfty}) are valid is described by the inequalities (see also Remark \ref{square}  below),
\begin{equation}\label{set}
0<\Re \sigma <\frac{1}{2} \Longleftrightarrow |\Im\alpha| < 2,\quad \sin{2\pi\eta} \neq 0,
\quad  \left|\arg \frac{\sin2\pi\eta}{\sin2\pi\sigma}\right| < \frac{\pi}{2} \Longleftrightarrow |\Im\nu| < \frac{1}{2},
\end{equation}
where $\sigma$ and $\eta$ are understood as functions of $\alpha$ and $\beta$  defined in (\ref{sigma}). We notice that this
set contains all sufficiently small pairs $(\alpha, \beta)$, all real pairs $(\alpha, \beta)$ such that the
corresponding $\eta$ satisfies the inequality $0< \eta < \frac{1}{2}$ (mod ($1$)) and all pure imaginary pairs such that $|\alpha| < 2$. In fact, it is convenient to
take $\sigma$ and $\eta$ as the independent  parameters and think about $\alpha$ and $\beta$ as their
functions, i.e.,
\begin{equation}\label{alphasigma}
\alpha = i(2 - 8\sigma),\quad
\beta = -\pi + 4\pi\eta -i(2 - 8\sigma)\ln 8 -2i \ln\frac{\Gamma(1-2\sigma)}{\Gamma(2\sigma)},
\end{equation}
where $\sigma$, $\eta$ are the complex numbers satisfying  (\ref{set}). The expressions of the asymptotic parameters at $x =\infty$
in terms of $\sigma$ and $\eta$ have already been presented in (\ref{connect}).

The derivation of  formulae  (\ref{connect}) is based
on the 
{ \it Isomonodromy-Riemann-Hilbert Method}. We again  refer the reader to monograph
\cite{FIKN} for more details and for general references  concerning the
connection problem for Painlev\'e equations.  In the framework of the Riemann-Hilbert method, the parameters $\sigma$ and $\eta$ have
an independent important meaning as the {\it monodromy data} of the auxiliary linear system
associated with the third Painlev\'e equation. This meaning of the parameters $\sigma$ and $\eta$
plays important role in the considerations of this paper, and it will be explained in detail in the next section.  

Equations (\ref{atzero}) and (\ref{atinfty}) in turn
imply the following behavior at zero and  at infinity of the corresponding tau-function (see also \cite{Jim}),
\begin{equation}\label{15}
\tau(x) = C_0x^{-\frac{\alpha^2}{8}} ( 1 + o(1)), \quad x\to 0,
\end{equation}
and
\begin{equation}\label{16}
\tau(x) = C_{\infty}x^{\nu^2}e^{\frac{x^2}{8} +2\nu x} ( 1 + o(1)), \quad x\to \infty.
\end{equation}
In fact, one can write a complete asymptotic series for the tau-function at both critical
points whose coefficients are explicit functions of the Cauchy data $\alpha$, $\beta$ or,
equivalently, of the  monodromy data $\sigma$, $\eta$. The issue which we are concerned with is the evaluation of the ratio 
\begin{equation}
\label{ratio}
 C_{\infty}/C_0
\end{equation}
in terms of the initial data $\alpha$, $\beta$. This  can not be done  just by using  the asymptotic equations (\ref{atzero}) - (\ref{atinfty}) and the connection
formulae (\ref{connect})-(\ref{sigma}). Indeed, we are dealing  here with the ``constant of integration'' problem.
For the special one-parameter family of solutions of equation (\ref{sineP3}) related to the Ising model, this problem
was solved by C. Tracy \cite{T} in 1991.  This special family is obtained by putting
\begin{equation}\label{tracy}
\eta =0 \quad\mbox{and}\quad  \sigma\in \mathbb{R}, \quad 0<\sigma < \frac{1}{2}.
\end{equation}
in (\ref{alphasigma}). Zero value of $\eta$ is excluded from set (\ref{set}) which means that 
the behavior of this special family   at infinity is very different from the oscillatory one given in (\ref{atinfty}).  In fact, all the solutions
from this family exponentially approach $\pi$  (mod$2\pi$),
$$
u(x) - \pi  \sim  i\kappa \sqrt{\frac{2}{\pi}}x^{-1/2}e^{-x}, \quad x \to \infty, \quad \kappa = -2\cos 2\pi\sigma.
$$

In his calculations, Tracy made use of the existence of an additional Fredholm determinant representation 
of the tau-function evaluated on the family (\ref{tracy}). We are interested  in a generic, two-parameter case where there is
no such representation. A conjectural answer to the problem has been produced  in \cite{ILT} with the help of the recently discovered
by O. Gamayun, N. Iorgov, and O. Lisovyy connection of the Painlev\'e tau-functions with the Virasoro conformal blocks \cite{Lis1}, \cite{Lis2}. In this paper we prove the conjecture of  \cite{ILT}. Our main result is the following theorem.
\begin{theorem}\label{the}Let $\sigma$ and $\eta $ be the ``monodromy'' parameters  of the Painlev\'e III function $u(x)$ 
 satisfying the inequalities (\ref{set}). Then the ratio  (\ref{ratio}) is given by
the formula,
$$
\dfrac{C_\infty}{C_0}=\dfrac{2^{\frac{3}{2}}e^{-i\frac{\pi}{4}}}{\pi (G(\frac{1}{2}))^4}
(2\pi)^{{i\nu}}2^{2\nu^2+{\sigma^2}{24}-12\sigma}e^{2\pi i(\eta^2-2\sigma\eta-\sigma^2+2\eta -\sigma) }
$$
\begin{equation}\label{answer}
\times\dfrac{\Gamma(1-2\sigma)}{\Gamma(2\sigma)}\left(\dfrac{G(1+i\nu)G(1+2\sigma)G(1-2\sigma){G}(1+\sigma+\eta+\frac{1-i\nu}{2}){G}(\frac{1-i\nu}{2}-\sigma-\eta)}{{G}(1+\sigma+\eta+\frac{1+i\nu}{2}){G}(\frac{1+i\nu}{2}-\sigma-\eta)}\right)^{{2}},
\end{equation}
where  $\nu$ is defined in  (\ref{connect}) and $G(z)$ is the Barnes $G$ - function.
\end{theorem}
It should be noticed that in \cite{ILT} a slightly different definition of the tau-function is used. The exact relation
of the constant (\ref{answer}) and the one conjectured in \cite{ILT} is discussed in the last section of the paper.

Our proof of Theorem \ref{the} is not based on   the conformal block connection. We use the Riemann-Hilbert representation
of the third Painlev\'e transcendent and the Malgrange-Bertola extension of the Jimbo-Miwa-Ueno definition of the tau-function. 

In the course of our proof, we also confirm one of the key observations
 of  \cite{ILT} that  ratio (\ref{answer}) determines the generating function of the canonical
 transformation of the canonical variables determined by the initial data $(\alpha, \beta)$ to the
 canonical variables determined by the  asymptotic data $(b_+, b_-)$ (see the end
 of Section 5 for more detail). In fact, this Hamiltonian interpretation of
 the pre-factors in the asymptotics of the Painlev\'e  tau-functions was first suggested in the work \cite{Lis3} of N.~Iorgov, O.~Lisovyy and Yu.~Tykhyy.

The  evaluation of  ratio (\ref{ratio}), which we have made rigorous in this paper,   is 
only one of  a series of highly nontrivial  predictions and already established facts which
came  from the remarkable discovery of Gamayun, Iorgov, and Lisovyy. These other
predictions and results, including the  key ingredient of the approach of \cite{Lis1} and \cite{Lis2},
which is the explicit conformal block series  representations for  the Painlev\'e tau-functions,
do not yet have their understanding in the framework of the Riemann-Hilbert method. 

We shall start the proof of Theorem \ref{the} with the reminding of the Isomonodromy-Riemann-Hilbert formalism for the Painlev\'e equation (\ref{sineP3})
(for more detail see, e.g., \cite{FIKN}). 
\section{The Riemann-Hilbert Representation of the Solutions of the Sine-Gordon/Painlev\'e III Equation}
The Riemann-Hilbert problem associated with equation (\ref{sineP3}) is posed  on the oriented  contour $\Gamma$
depicted in Figure \ref{fig1}, and it consists in the finding of a $2\times 2$ matrix-valued function $\Psi(\lambda)$ which satisfies the following properties.
\begin{figure}[h]
\centering
{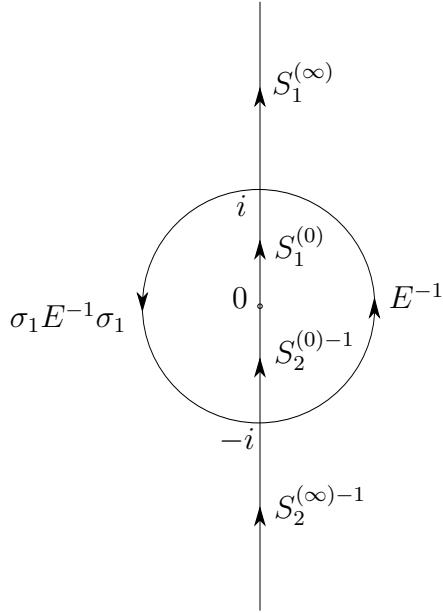}
\caption{Contour $\Gamma$}
\label{fig1}
\end{figure}
\begin{itemize}
\item The function $\Psi(\lambda)$ is analytic on $\mathbb{C} \setminus\{\Gamma\}$, it has continuous $\pm$ - limits on the contour $\Gamma$, and these limits satisfy jump condition $\Psi_+(\lambda)=\Psi_{-}(\lambda)S(\lambda)$. Here $``+"$ denotes the boundary values from the left side of the contour 
and $``-"$ denotes the boundary values from the right side of the contour.  Jump matrix $S(\lambda)$ is piecewise constant,
its different components are indicated in Figure \ref{fig1},  and  they are  given by the equations,
\begin{equation}\label{RHdata1}
S_1^{(\infty)}=S_2^{(0)}=\left(
\begin{array}{cc}
1&0\\
p+q&1\\
\end{array}
\right),\quad S_2^{(\infty)}=S_1^{(0)}=\left(
\begin{array}{cc}
1&p+q\\
0&1\\
\end{array}
\right),
\end{equation}
\begin{equation}\label{RHdata2}
E=\dfrac{1}{\sqrt{1+pq}}\left(
\begin{array}{cc}
1&p\\
-q&1\\
\end{array}
\right),\quad p, q \in \mathbb{C}, \quad 1+pq \neq 0,\quad \sigma_1=\left(
\begin{array}{cc}
0&1\\
1&0\\
\end{array}
\right). 
\end{equation}
\item The function $\Psi(\lambda)$ satisfies the following conditions at zero and infinity
\begin{equation}\label{pzero}
\Psi(\lambda)=P_0(I+M_1^{(0)}\lambda+O(\lambda^2))e^{-\frac{i}{\lambda}\sigma_3},\quad \lambda\to 0,\nonumber
\end{equation}
\begin{equation}\label{pinfin}
\Psi(\lambda)=\Bigl(I+\dfrac{M_1^{(\infty)}}{\lambda}+O\Bigl(\dfrac{1}{\lambda^2}\Bigr)\Bigr)e^{-\frac{ix^2\lambda}{16}\sigma_3},\quad \lambda\to \infty,
\end{equation}
where $P_0, M_1^{(0)}, M_1^{(\infty)}$ here are some constant in $\lambda$ matrices and
$$
\sigma_3=\left(
\begin{array}{cc}
1&0\\
0&-1\\
\end{array}
\right).
$$
\end{itemize}
The Riemann-Hilbert problem is uniquely and meromorphically in $x$ solvable for all  \linebreak $p, q \in \mathbb{C}, \quad 1+pq \neq 0$ \cite{Niles} and the
corresponding  solution $u(x) \equiv u(x; p, q)$ of the third Painlev\'e equation (\ref{sineP3}) is given by the formula,
$$
u(x) = 2\arccos (P_0)_{11}.
$$
In fact, the following equation takes place,
\begin{equation}\label{P0def}
P_0=\left(\begin{array}{cc}
\mathrm{cos}(\frac{u}{2})&-i\mathrm{sin}(\frac{u}{2})\\
-i\mathrm{sin}(\frac{u}{2})& \mathrm{cos}(\frac{u}{2})
\end{array}\right) \equiv e^{-\frac{iu\sigma_1}{2}}.
\end{equation}
\begin{remark} The above Riemann-Hilbert setting corresponds to the generic solutions of (\ref{sineP3}). There is one parameter family
of a separatrix solution which is characterized by the following Riemann-Hilbert data
\begin{equation}\label{RHdata3}
S_1^{(\infty)}=S_2^{(0)}=\left(
\begin{array}{cc}
1&0\\
\kappa&1\\
\end{array}
\right),\quad S_2^{(\infty)}=S_1^{(0)}=\left(
\begin{array}{cc}
1&\kappa\\
0&1\\
\end{array}
\right),\quad 
E=\pm i\left(
\begin{array}{cc}
0&1\\
1&0\\
\end{array}
\right),\quad \kappa \in \mathbb{C}.\nonumber
\end{equation}
This is the family which includes the McCoy-Tracy-Wu solution (\ref{tracy}) and which is not considered 
in this paper. As it has already been mentioned, the constant problem for this family was solved in \cite{T}.
\end{remark}
The parameters $p,q\in\mathbb{C}$  in (\ref{RHdata1}), (\ref{RHdata2}) are connected with the parameters of asymptotic of $u(x)$ via
$$ p:=-i\dfrac{\sin 2\pi (\sigma + \eta)}{\sin 2\pi\eta},\quad q:=i\dfrac{\sin 2\pi (\sigma - \eta)}{\sin 2\pi\eta}.$$
Conditions $|\Im \alpha|<2$ and $|\Im \nu|<1/2$ can be rewritten in terms of $p$ and $q$ as 
\begin{equation}\label{condpq}
p+q\notin (-i\infty, -2i]\cup[2i,+i\infty),\quad\mbox{and}\quad  pq\notin (-\infty,-1],
\end{equation}
respectively.
\begin{remark}\label{square}  Conditions (\ref{condpq}) are the conditions which appear during the asymptotic
analysis of the Riemann-Hilbert problem. The asymptotic parameter $\nu$ is related to $p$, $q$ according
to the equation,
$$
\nu = -\frac{1}{2\pi}\ln(1 + pq).
$$
Also, 
$$
1+pq = \frac{\sin^{2}2\pi\sigma}{\sin^{2}2\pi\eta}.
$$
We restrict ourselves in (\ref{set}) to the inequality $\left|\arg \frac{\sin2\pi\eta}{\sin2\pi\sigma}\right| <\pi/2$,
instead of the inequality $\left|\arg \left(\frac{\sin2\pi\eta}{\sin2\pi\sigma}\right)^2\right| <\pi$ by a technical reason.
This means that we actually analyze one of the components of the full set of the
Cauchy data corresponding to the generic asymptotic behavior (\ref{atzero}) and (\ref{atinfty}).
The another component is defined by the condition,
$$\left|\arg \frac{\sin2\pi\eta}{\sin2\pi\sigma} - \pi \right| <\pi/2,$$
which in turn implies the following change in formulae (\ref{connect}),
$$
 e^{\pi\nu} = -\frac{\sin2\pi\eta}{\sin2\pi\sigma}.
 $$
The analysis presented in this paper can be easily extended on this component of initial data as well. Actually, the $\tau$-function does not change if we add $2\pi i $ to the function $u(x)$. But $\eta$ is shifted by $\frac{1}{2}$. So such change of variable allows us to go from one component to another.
\end{remark}
Function $\Psi(\lambda)$ satisfies system of linear ordinary differential equations
\begin{equation}\label{system}
\dfrac{d\Psi}{d \lambda}=A(\lambda)\Psi(\lambda),
\end{equation}
\[
\dfrac{d\Psi}{d x}=U(\lambda)\Psi(\lambda),
\]
\begin{equation}\label{A}
A(\lambda)=-\dfrac{ix^2\sigma_3}{16}-\frac{ix u_x\sigma_1}{4\lambda}+\dfrac{P_0(i\sigma_3)P^{-1}_0}{\lambda^2},
\end{equation}
\[
U(\lambda)=-\dfrac{i\lambda x \sigma_3}{8}-\dfrac{i u_x \sigma_1}{2}.
\]
Equation \eqref{sineP3} is the compatibility condition for this system and it describes isomonodromic deformations of the system \eqref{system}. From this point of view $\sigma$ and $\eta$ play role of the  monodromy data. 

We complete this overview  of the Riemann-Hilbert formalism for equation (\ref{sineP3}) by presenting
the general alternative definition of the  Jimbo-Miwa-Ueno tau-function  in terms of the solution $\Psi(\lambda)$ 
of the Riemann-Hilbert problem.

Define
$$
\hat{\Psi}^{(\infty)}(\lambda):=\Psi(\lambda)e^{\frac{ix^2\lambda}{16}\sigma_3},\quad |\lambda|>R.
$$
Then, according to \cite{JM} the equation, 
\begin{equation}\label{taudef2}
\omega_{JMU}=-\mathop{res}\limits_{\lambda=\infty} \mathrm{Tr}\Bigl((\hat{\Psi}
^{(\infty)}(\lambda))^{-1}({\hat{\Psi}^{(\infty)}}(\lambda))'(-\frac{i\lambda x}
{8}\sigma_3)\Bigr)dx,
\end{equation} 
defines the differential form whose antiderivative is the logarithm of tau-function.
Actually, from \eqref{pinfin} we have that
\begin{equation}\label{phin}
\hat{\Psi}^{(\infty)}(\lambda)=I+\dfrac{M_1^{(\infty)}}{\lambda}+O\Bigl(\dfrac{1}{\lambda^2}\Bigr).
\end{equation}

 Substituting \eqref{pinfin} to the equation \eqref{system}, one can express $M_1^{(\infty)}$ in terms of $u(x)$ (see \cite{FIKN},\cite{Niles}).
 \begin{equation}\label{m1in}
M_1^{(\infty)}=-\frac{2ixu_x}{x^2}\sigma_2-i\left(\cos u -\frac{u_x^2}{2}\right)\sigma_3,
\quad
\sigma_2=\left(
\begin{array}{cc}
0&-i\\
i&0\\
\end{array}
\right). 
\end{equation}
Substituting, in turn,  \eqref{phin}-\eqref{m1in} into \eqref{taudef2} yields the relation,  $\omega_{JMU} = d\ln\tau(x)$,
 where $d\ln\tau$ is defined in (\ref{taudef}). In other words, the tau-function  can be alternatively defined as 
\begin{equation}\label{taudef3}
\tau \equiv \tau_{JMU}(x,p,q) = e^{\intop\omega_{JMU}}.\nonumber
\end{equation}
So defined the tau-function is unique up to multiplication by a constant depending on $p$ and $q$. 
The key fact for us is that, following \cite{B}, it is possible to extend Jimbo-Miwa-Ueno differential form (\ref{taudef2}) on vector fields in $p$ and $q$ in such a way, that it will remains a closed form. Such extension will allow us to define tau-function already up to a constant, which does not depend on $p$ and $q$.

\section{Malgrange-Bertola Differential Form } 
 In this section we basically repeat the calculations and
the results of Section 5.1 of paper \cite{B} adjusting them to our special case.

Put $Y(\lambda)=\Psi(\lambda)e^{(\frac{ix^2\lambda}{16}+\frac{i}{\lambda})\sigma_3}$. Denote $G(\lambda)$ the jump matrix for $Y(\lambda)$. 
Following \cite{B},\cite{Mal}, we define   the {\it Malgrange-Bertola differential form} by the equation
\begin{equation}\label{O}
\omega_{MB}[\partial]=\intop_{{\Gamma}} \mathrm{Tr}(Y_{-}^{-1}Y'_{-}(\partial G)G^{-1} \dfrac{d\lambda}{2\pi i}.
\end{equation}
Here $\partial$ denotes the vector field in the space of parameters $x,p,q$, and the prime denotes derivative with respect to $\lambda$. This differential form was introduced originally by B. Malgrange in \cite{Mal} for the case
when the contour $\Gamma$ is a circle. M. Bertola in \cite{B} has extended the Malgrange's definition
to an arbitrary Riemann-Hilbert setting.

Let us establish the connection of this form with Jimbo-Miwa-Ueno form. For the case of general Riemann-Hilbert problem, 
the analog of this Lemma was proven in \cite{B}.
\begin{lemma}\label{l1}The Malgrange-Bertola differential form, evaluated on the vector fields in parameter $x$, is equal to the Jimbo-Miwa-Ueno form up to a term, depending only on $G(\lambda)$.
\begin{equation}
\label{l1e}
\omega_{MB}[\partial_x]=\omega_{JMU}[\partial_x]-\left[\intop_{{\Gamma}}\mathrm{Tr}\Bigl(G^{-1}G'\dfrac{ix\lambda}{8}\sigma_3\Bigr) \dfrac{d\lambda}{2\pi i}\right]-\frac{x}{4}.
\end{equation}
\end{lemma}
\begin{proof}

First, we have
$$
G(\lambda)=e^{-(\frac{ix^2\lambda}{16}+\frac{i}{\lambda})\sigma_3}S(\lambda)e^{(\frac{ix^2\lambda}{16}+\frac{i}{\lambda})\sigma_3},
$$
where $S(\lambda)$ is the jump matrix for $\Psi(\lambda)$. Hence,
\begin{equation}\label{dg}
\partial_x G(\lambda) G^{-1}(\lambda)=-\dfrac{ix\lambda}{8}\sigma_3+ \dfrac{ix\lambda}{8}G(\lambda)\sigma_3 G^{-1}(\lambda).
\end{equation} 
Also, we have that
\begin{equation}\label{y}
Y_{+}(\lambda)=Y_{-}(\lambda)G(\lambda),
\end{equation}
and, by $\lambda$ - differentiation, 
\begin{equation}\label{dy}
Y_{+}'(\lambda)=Y_{-}'
(\lambda)G(\lambda)+Y_{-}(\lambda)G'(\lambda).
\end{equation}

Substituting \eqref{dg},\eqref{y},\eqref{dy} in \eqref{O} we get
$$
\omega_{MB}[\partial_x]=\left[\intop_{{\Gamma}} \mathrm{Tr}\Bigl((Y_{+}^{-1}Y'_{+}-Y_{-}^{-1}Y'_{-})\dfrac{ix\lambda}{8}\sigma_3-G^{-1}G'\dfrac{ix\lambda}{8}\sigma_3\Bigr) \dfrac{d\lambda}{2\pi i}\right].
$$

Let us introduce the  notation for the parts of the contour $\Gamma$ as it is indicated in Figure \ref{fig2}.
\begin{figure}[h]
\centering
{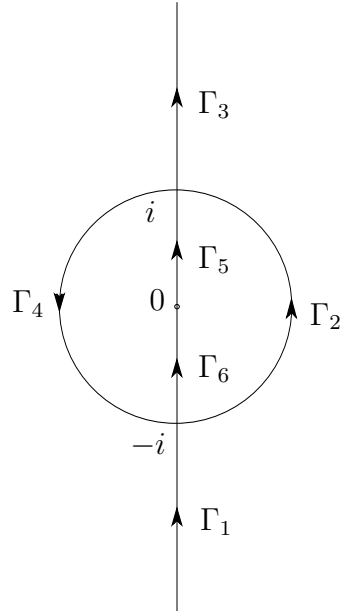}
\caption{Parts of the contour ${\Gamma}$}
\label{fig2}
\end{figure}
\newpage
We have 
$$
\intop_{\Gamma_5\cup\Gamma_6\cup\Gamma_4} \mathrm{Tr}\Bigl(Y_{+}^{-1}Y'_{+}\dfrac{ix\lambda}{8}\sigma_3\Bigr) \dfrac{d\lambda}{2\pi i}=0,
$$
$$
\intop_{\Gamma_2} \mathrm{Tr}\Bigl(Y_{+}^{-1}Y'_{+}\dfrac{ix\lambda}{8}\sigma_3\Bigr) \dfrac{d\lambda}{2\pi i}-\intop_{\Gamma_5\cup\Gamma_6} \mathrm{Tr}\Bigl(Y_{-}^{-1}Y'_{-}\dfrac{ix\lambda}{8}\sigma_3\Bigr) \dfrac{d\lambda}{2\pi i}=0.
$$
We also have 
$$
Y(\lambda)=\hat{\Psi}^{(\infty)}(\lambda)e^{\frac{i}{\lambda}\sigma_3},\quad |\lambda|>R.
$$
Hence,
$$
\intop_{\Gamma_1\cup\Gamma_3} \mathrm{Tr}\Bigl(Y_{+}^{-1}Y'_{+}\dfrac{ix\lambda}{8}\sigma_3\Bigr) \dfrac{d\lambda}{2\pi i}-\intop_{\Gamma_1\cup\Gamma_2\cup\Gamma_3\cup\Gamma_4} \mathrm{Tr}\Bigl(Y_{-}^{-1}Y'_{-}\dfrac{ix\lambda}{8}\sigma_3\Bigr) \dfrac{d\lambda}{2\pi i}
$$
$$
=\mathop{\mathrm{res}}\limits_{\lambda=\infty}\mathrm{Tr}\Bigl(e^{-\frac{i}{\lambda}\sigma_3}(\hat{\Psi}^{(\infty)}(\lambda))^{-1}(\hat{\Psi}^{(\infty)}(\lambda))'e^{\frac{i}{\lambda}\sigma_3}\dfrac{ix\lambda}{8}\sigma_3-e^{-\frac{i}{\lambda}\sigma_3}(\hat{\Psi}^{(\infty)}(\lambda))^{-1}(\hat{\Psi}^{(\infty)}(\lambda)){\frac{i}{\lambda^2}\sigma_3}e^{\frac{i}{\lambda}\sigma_3}\dfrac{ix\lambda}{8}\sigma_3 \Bigr)
$$
$$
=\mathop{\mathrm{res}}\limits_{\lambda=\infty}\mathrm{Tr}\Bigl((\hat{\Psi}^{(\infty)}(\lambda))^{-1}(\hat{\Psi}^{(\infty)}(\lambda))'\dfrac{ix\lambda}{8}\sigma_3\Bigr)-\frac{x}{4},
$$
and \eqref{l1e} follows.
\end{proof}
This lemma means, that Malgrange-Bertola form is indeed a good candidate for extension of Jimbo-Miwa-Ueno form. 
However,  there is an additional term, depending only on $G(\lambda)$.
One can, following again \cite{B}, cancel it considering the {\it modified Malgrange-Bertola} form $\omega=\omega_{MB}+\theta$, where
$$\theta[\partial]=\dfrac{1}{2}\intop_{\hat{\Gamma}}\mathrm{Tr}(G'G^{-1}(\partial G)G^{-1}) \dfrac{d\lambda}{2\pi i}.$$
In the notations of \cite{B} $\omega$ is the form $\Omega$ from  Definition 2.2 of \cite{B}.

We have $\omega[\partial_x]=\omega_{JMU}[\partial_x]-\frac{x}{4}$. Indeed,
$$
G^{-1}G'\dfrac{ix\lambda}{8}\sigma_3=\dfrac{ix\lambda}{8}\left(\dfrac{ix^2}{16}-\dfrac{i}{\lambda^2}\right)(I-G^{-1}\sigma_3 G \sigma_3),
$$
$$
\mathrm{Tr}(G'G^{-1}(\partial_x G)G^{-1})=\dfrac{ix\lambda}{8}\left(\dfrac{ix^2}{16}-\dfrac{i}{\lambda^2}\right)2\mathrm{Tr}(I-G^{-1}\sigma_3 G \sigma_3),
$$
and the additional term depending only on $G(\lambda)$ cancels. 

In the next section we will express the form $\omega$ in terms of the coefficients of the
asymptotic expansions of $Y(\lambda)$ at $\lambda = 0$ and  at $\lambda = \infty$.
We call  this expression a ``localization'' of the original integral formula (\ref{O}) for 
the Malgrange-Bertola form. This localized version will simplify dramatically the
further analysis of the form $\omega$. 
\section{Localization}
Let us  introduce the function 
$$
{\Theta}(\lambda)=\partial Y(\lambda)Y(\lambda)^{-1},
$$
where $\partial$ means the differentiation with respect to  one of the three
parameters, $x$, $p$, or $q$. The $\partial$-version of equation \eqref{dy}
reads,
$$
\partial Y_+(\lambda) = \partial Y_-(\lambda)G(\lambda) + Y_-(\lambda)\partial G(\lambda).
$$
Expressing $G(\lambda)$ from \eqref{y} as $G(\lambda) = Y^{-1}_-(\lambda)Y_+(\lambda)$ we rewrite the last
equations as
$$
\partial Y_+(\lambda)Y^{-1}_+(\lambda) = \partial Y_-(\lambda)Y^{-1}_-(\lambda)+ Y_-(\lambda)\partial G(\lambda)Y^{-1}_+(\lambda),
$$
or as
\begin{equation}\label{jumppartial}
\partial Y_+(\lambda)Y^{-1}_+(\lambda) - \partial Y_-(\lambda)Y^{-1}_-(\lambda) =Y_-(\lambda)\partial G(\lambda)Y^{-1}_+(\lambda).
\end{equation}
The Sokhotski-Plemelj formula would then imply (cf. Lemma 2.1 of \cite{B}) that
\begin{equation}\label{f1}
\Theta(\lambda)=\intop_{\Gamma}\dfrac{Y_{-}(y)\partial G(y)Y_{+}^{-1}(y)}{y-\lambda}\dfrac{dy}{2\pi i}.
\end{equation}
Substituting $Y(\lambda)=\Psi(\lambda)e^{(\frac{ix^2\lambda}{16}+\frac{i}{\lambda})\sigma_3}$ in \eqref{system}, we have
\begin{equation}\label{ysystem}
Y'(\lambda)=A(\lambda)Y(\lambda)+\Bigl(\dfrac{ix^2}{16}-\dfrac{i}{\lambda^2}\Bigr)Y(\lambda)\sigma_3,
\end{equation}
and
$$
\omega_{MB}[\partial]=\intop_{{\Gamma}} \mathrm{Tr}(Y_{-}^{-1}Y'_{-}(\partial G)G^{-1}) \dfrac{d\lambda}{2\pi i}=\intop_{{\Gamma}} \Bigl(\dfrac{ix^2}{16}-\dfrac{i}{\lambda^2}\Bigr)\mathrm{Tr}(\sigma_3(\partial G)G^{-1}) \dfrac{d\lambda}{2\pi i}
$$
\begin{equation}\label{loc1}
+\intop_{{\Gamma}} \mathrm{Tr}(AY_{-}(\partial G)Y_{+}^{-1}) \dfrac{d\lambda}{2\pi i}.
\end{equation}
We introduce notation
\begin{equation}\label{yzero}
Y(\lambda)=P_0(I+\stackrel{\circ}{m_1}\lambda+O(\lambda^2)), \quad \lambda\to 0,
\end{equation}
\begin{equation}\label{yinfin}
Y(\lambda)=\Bigl(I+\dfrac{m_1^{(\infty)}}{\lambda}+O\Bigl(\dfrac{1}{\lambda^2}\Bigr)\Bigr), \quad \lambda\to \infty.
\end{equation}
Substituting these expressions for $Y(\lambda)$ to the definition of $\Theta(\lambda)$ we get
$$
\Theta(\lambda)=(\partial P_0)P_0^{-1}+P_0(\partial\stackrel{\circ}{m_1})P_0^{-1}\lambda+O(\lambda^2), \quad \lambda\to 0,
$$
$$
\Theta(\lambda)=\dfrac{\partial m_1^{(\infty)}}{\lambda}+O\Bigl(\dfrac{1}{\lambda^2}\Bigr), \quad \lambda\to \infty.
$$
Comparing these formulae with \eqref{f1} we arrive at the relations
\begin{equation}\label{loc2}
\intop_{\Gamma}{Y_{-}(y)\partial G(y)Y_{+}^{-1}(y)}\dfrac{dy}{2\pi i}=-\partial m_1^{(\infty)},
\end{equation}
\begin{equation}\label{loc3}
\intop_{\Gamma}\dfrac{Y_{-}(y)\partial G(y)Y_{+}^{-1}(y)}{y}\dfrac{dy}{2\pi i}=(\partial P_0)P_0^{-1},
\end{equation}
\begin{equation}\label{loc4}
\intop_{\Gamma}\dfrac{Y_{-}(y)\partial G(y)Y_{+}^{-1}(y)}{y^2}\dfrac{dy}{2\pi i}=P_0(\partial\stackrel{\circ}{m_1})P_0^{-1}.
\end{equation}
Let us look now at the last integral in equation (\ref{loc1}). Putting in it formula (\ref{A}) for $A(\lambda)$, we will see
that this integral can be re-written as 
$$
\intop_{{\Gamma}} \mathrm{Tr}\Bigl(AY_{-}(\partial G)Y_{+}^{-1}\Bigr) \dfrac{d\lambda}{2\pi i}
= -\frac{ix^2}{16}\intop_{{\Gamma}} \mathrm{Tr}\Bigl(\sigma_3Y_{-}(\partial G)Y_{+}^{-1}\Bigr) \dfrac{d\lambda}{2\pi i}
$$
$$
-\frac{ixu_x}{4}\intop_{{\Gamma}} \mathrm{Tr}\left(\sigma_1\frac{Y_{-}(\partial G)Y_{+}^{-1}}{\lambda}\right) \dfrac{d\lambda}{2\pi i}
+ \intop_{{\Gamma}} \mathrm{Tr}\left(P_0(i\sigma_3)P^{-1}_0\frac{Y_{-}(\partial G)Y_{+}^{-1}}{\lambda^2}\right) \dfrac{d\lambda}{2\pi i}
$$
$$
= -\frac{ix^2}{16}\mathrm{Tr}\Bigl(\sigma_3\intop_{{\Gamma}} Y_{-}(\partial G)Y_{+}^{-1} \dfrac{d\lambda}{2\pi i}\Bigr)
$$
$$
-\frac{ixu_x}{4} \mathrm{Tr}\left(\sigma_1\intop_{{\Gamma}}\frac{Y_{-}(\partial G)Y_{+}^{-1}}{\lambda} \dfrac{d\lambda}{2\pi i}\right)
+ \mathrm{Tr}\left(P_0(i\sigma_3)P^{-1}_0\intop_{{\Gamma}} \frac{Y_{-}(\partial G)Y_{+}^{-1}}{\lambda^2} \dfrac{d\lambda}{2\pi i}\right).
$$
The last equation, with the help of (\ref{loc2}) - (\ref{loc4}), is transformed into the {\it localized} formula,
\begin{equation}\label{loc5}
\intop_{{\Gamma}} \mathrm{Tr}\Bigl(AY_{-}(\partial G)Y_{+}^{-1}\Bigr) \dfrac{d\lambda}{2\pi i}=
\frac{ix^2}{16}\mathrm{Tr}\Bigl(\sigma_3\partial m^{(\infty)}_1\Bigr)
-\frac{ixu_x}{4} \mathrm{Tr}\left(\sigma_1(\partial P_0)P^{-1}_0\right)
+ i \mathrm{Tr}\left(\sigma_3\partial {\stackrel{\circ}{m_1}}\right).
\end{equation}
Substituting the  derivative of $G$ with respect to $\lambda$ in the formula for $\theta$, we have that
$$
\theta[\partial]=\dfrac{1}{2}\intop_{\Gamma}\Bigl(\dfrac{ix^2}{16}-\dfrac{i}{\lambda^2}\Bigr)\mathrm{Tr}\Bigl(\sigma_3(G^{-1}\partial G-(\partial G) G^{-1})\Bigr)\dfrac{d\lambda}{2\pi i}.
$$
Together with (\ref{loc5}) this  gives us the following formula for $\omega$
$$
\omega[\partial]=\dfrac{1}{2}\intop_{{\Gamma}} \left(\dfrac{ix^2}{16}-\dfrac{i}{\lambda^2}\right)\mathrm{Tr}\Bigl(\sigma_3((\partial G)G^{-1}+G^{-1}\partial G)\Bigr) \dfrac{d\lambda}{2\pi i}
$$
$$
+ \dfrac{ix^2}{16}\mathrm{Tr}\left(\sigma_3\partial m_1^{(\infty)}\right)-\frac{ix u_x}{4}\mathrm{Tr}\left(\sigma_1(\partial P_0)P_0^{-1}\right)
+i\mathrm{Tr}\left(\sigma_3\partial\stackrel{\circ}{m_1}\right).
$$
One can check directly  that
$$
\mathrm{Tr}\Bigl(\sigma_3((\partial G)G^{-1}+G^{-1}(\partial G))\Bigr)\equiv 0.
$$
Therefore, we finally have that
\begin{equation}\label{loc6}
\omega[\partial] = \dfrac{ix^2}{16}\mathrm{Tr}\left(\sigma_3\partial m_1^{(\infty)}\right)-\frac{ix u_x}{4}\mathrm{Tr}\left(\sigma_1(\partial P_0)P_0^{-1}\right)
+i\mathrm{Tr}\left(\sigma_3\partial\stackrel{\circ}{m_1}\right).
\end{equation}
Substituting \eqref{yzero}, \eqref{yinfin} to the equation \eqref{ysystem}, one can express the coefficients of asymptotics of $Y(\lambda)$
at $\lambda =0$ and at $\lambda=\infty$ in terms of $u$. In particular, one gets (cf. \cite{Niles}),
$$
m_1^{(\infty)}=-\frac{2ixu_x}{x^2}\sigma_2-i\left(\cos u -1-\frac{u_x^2}{2}\right)\sigma_3,
$$
$$
\stackrel{\circ}{m_1}=\frac{ixu_x}{8}\sigma_2-i\left(\frac{x^2}{16}(\cos u -1)-\frac{x^2u_x^2}{32}\right)\sigma_3.
$$
Inserting  these equations together with  formula (\ref{P0def})  in (\ref{loc6})  and using also the fact, that $u(x)$ satisfies \eqref{sineP3} 
we transform equation (\ref{loc6}) into the  final expression for the form $\omega$ in terms of 
$u$ and its derivatives with respect to $x$, $p$ and $q$.
\begin{proposition}\label{prop} The modified Malgrange-Bertola differential form $\omega$ admits the
following representation 
$$
\omega=\left(-\frac{x u_x^2}{8}+\frac{x}{4}(\cos u-1)\right)dx-\left(\frac{x^2}{4}u_p \sin u+\frac{x^2}{4}u_x u_{px}+\dfrac{x u_x u_p}{4}\right)dp
$$
\begin{equation}\label{omegaloc}
-\left(\frac{x^2}{4}u_q \sin u+\frac{x^2}{4}u_x u_{qx}+\dfrac{x u_x u_q}{4}\right)dq.
\end{equation}
\end{proposition}
We notice that from (\ref{omegaloc}) we have again the statement of Lemma \ref{l1}, that  is that $\omega[\partial_x] = \partial_x \ln \tau - \frac{x}{4}$.
We want to mention again that this part of the localization formulae has already been obtained in \cite{B}.  We also want to emphasize the important role
which is played by the $\lambda$ - equation \eqref{system} in the derivation of the $p, q$ - part of equation \eqref{omegaloc}. It is the
use of this equation that allowed us to present the original Malgrange-Bertola integral \eqref{O}, first in the form \eqref{loc1}, and
then in the localized form \eqref{loc5}. In fact, similar technique has already been used in the study of Toeplitz determinants
with the Fisher-Hartwig singularities in paper \cite{DIK}  - see Appendix 6 and Lemma 6.2 of that paper.
\begin{remark} As it was pointed out to the authors by M. Bertola, equation (\ref{jumppartial}) can be used in the derivations
of this section one more time and help to make a significant short cut from equation (\ref{loc1}) to the localized
form (\ref{loc5}). Indeed, Bertola's suggestion is to use relation (\ref{jumppartial}) for the product
$Y_-(\lambda)G(\lambda)Y_+^{-1}(\lambda)$ in the last integral of (\ref{loc1}) directly and
rewrite this integral as
$$
\intop_{{\Gamma}} \mathrm{Tr}\Bigl(A(\lambda)Y_-(\lambda)\partial G(\lambda)Y^{-1}_+(\lambda)\Bigr) \dfrac{d\lambda}{2\pi i} =\int_{\Gamma}\mathrm{Tr}\,\left(A(\lambda)\Bigl( \partial Y_+(\lambda)Y^{-1}_+(\lambda) - \partial Y_-(\lambda)Y^{-1}_-(\lambda) \Bigr)\right)
\frac{d\lambda}{2\pi i}
$$
\begin{equation}\label{bertola}
= \sum_{poles\,of\,A(\lambda)d\lambda}\mathop{res}\,\mathrm{Tr}\,\Bigl(A(\lambda) \partial Y(\lambda)Y^{-1}(\lambda)\Bigr).
\end{equation}
This is a quite general construction which allows one to localize the Malgrange-Bertola form for an arbitrary isomonodromic 
Riemann-Hilbert problem. In our  case, one has to evaluate the, properly understood,  residues at the points $\lambda=0, \infty$.
The result will be equation (\ref{loc5}).
\end{remark}
\section{Proof of Theorem \ref{the}}
Let us compute $d\omega$. First we have,
$$
d\left[\left(\dfrac{x u_x^2}{8}-\dfrac{x}{4}(\cos u-1)\right)dx\right]=\left(\dfrac{xu_xu_{px}}{4}+\dfrac{xu_p\sin u }{4}\right)
dp\wedge dx+\left(\dfrac{xu_xu_{qx}}{4}+\dfrac{x u_q\sin u }{4}\right)dq\wedge dx.
$$
Then, using the fact that $u(x)$ satisfies equation \eqref{sineP3}, we get that
$$
d\left[\left(\dfrac{x^2}{4}u_p \sin u+\dfrac{x^2}{4}u_x u_{px}+\dfrac{x u_x u_p}{4}\right)dp\right]=
\left(\dfrac{xu_xu_{px}}{4}+\dfrac{xu_p\sin u }{4}\right)dx\wedge dp$$
$$
+\left(\dfrac{x^2}{4}u_{pq}\sin u+
\dfrac{x^2}{4}u_p u_q \cos u+ \dfrac{x^2}{4}u_{px} u_{qx}+\dfrac{x^2}{4}u_x u_{xpq}+\dfrac{x}{4}u_x u_{pq}
+\dfrac{x}{4}u_p u_{qx}\right)dq\wedge dp.
$$
and
$$
d\left[\left(\dfrac{x^2}{4}u_q \sin u+\dfrac{x^2}{4}u_x u_{qx}+\dfrac{x u_x u_q}{4}\right)dq\right]=
\left(\dfrac{xu_xu_{qx}}{4}+\dfrac{xu_q\sin u }{4}\right)dx\wedge dq$$
$$
+\left(\dfrac{x^2}{4}u_{pq}\sin u+
\dfrac{x^2}{4}u_p u_q \cos u+ \dfrac{x^2}{4}u_{px} u_{qx}+\dfrac{x^2}{4}u_x u_{xpq}+\dfrac{x}{4}u_x u_{pq}
+\dfrac{x}{4}u_q u_{px}\right)dp\wedge dq.
$$
Adding up the last three equations we obtain that
 $$
 d\omega=\frac {v_p u_q - v_qu_p}{4}dq\wedge dp,
 $$ 
 where $v=xu_x$. From equation (\ref{sineP3}) it follows that
 $$
 \frac{d}{dx}(v_p u_q - v_qu_p) = 0,
 $$
 and hence we can observe that
 $$
 d\omega = \lim_{x\to 0}d\omega = \frac {\alpha_p \beta_q - \alpha_q\beta_p}{4}dq\wedge dp = \frac{d\beta\wedge d\alpha}{4}.
 $$
Therefore, if we define 
 $$
w=\omega+\frac{x}{4}dx+\frac{\alpha d\beta}{4},
$$
then the form $w$  will be a closed form on the full set of parameters, $(x, p, q)$ and such that $w[\partial_x]=w_{JMU}[\partial_x]$. 
This means we can put,
$$
\tau=e^{\intop w},
$$
and this equation would define the tau-function up to a constant, which does not depend on $p$ and $q$.
\begin{remark} It is worth noticing, that from our analysis it follows that, in the case of the
Painlev\'e III equation (\ref{sineP3}), the external differential of the (modified) 
Malgrange-Bertola form $\omega$ is proportional  to the canonical symplectic form (\ref{symp})  for the Hamiltonian dynamics of the Painlev\'{e} equation; indeed, we have that 
 i.e.,
$$
 d\omega =- \frac{1}{4}\Omega.
 $$
\end{remark}
The next step is to evaluate the small and the large $x$  asymptotics of the form $w$. To this end we shall
use asymptotics \eqref{atzero},\eqref{atinfty} for $u(x)$ and  make the following temporary  technical assumptions,
\begin{equation}\label{techass}
|\Im \nu| < \frac{1}{6},\quad |\Im\alpha| < 1.
\end{equation}
In our calculations we will need more terms of the large $x$ asymptotics at infinity which are given
in \cite{ILT},  
$$
u(x)=b_+{{e}^{ix}}{x}^{i\nu-\frac{1}{2}}+b_-{{e}^{-ix}}{x}^{-i\nu-\frac{1}{2}}
$$
$$
+\frac{i b_+}{8} ( 6{\nu}^{2}+4i\nu-1 ) {
{e}^{ix}}{x}^{i\nu-\frac{3}{2}}-\frac{i b_-}{8}( 6{\nu}^{2}-4i\nu-1 ) {{e}^{
-ix}}{x}^{-i\nu-\frac{3}{2}}
$$
\begin{equation}\label{inftyext}
-\frac{1}{48}{b^{3}_+}{
{ e}^{3 ix}}{x}^{3i\nu-\frac{3}{2}}-\frac{1}{48}{b^{3}_-}{{e}^{-3ix}}{x}^
{-3i\nu-\frac{3}{2}} + O(x^{-\frac{5}{2}+5|{\Im \nu}|}).
\end{equation}
Substituting this asymptotics at the right hand side of equation (\ref{omegaloc}), we shall arrive, after rather tedious though
straightforward calculations,  at the following
asymptotic representation of the form $\omega$ as $x \to \infty$,
$$
\omega=d(2\nu x+{\nu^2}{\ln x}+\nu^2)-\frac{i}{4}(b_+db_{-}-b_-db_{+})+\left(\frac{ib_+^2}{8}e^{2ix}x^{2i\nu-1}-\frac{ib_-^2}{8}e^{-2ix}x^{-2i\nu-1}\right)dx
$$
\begin{equation}\label{omegainfty}
+ O(x^{-2+6|\mathrm{\Im}\nu|})dx+O(x^{-1+6|\mathrm{\Im \nu}|}))dp+
O(x^{-1+6|\mathrm{\Im \nu}|}))dq,\quad x\to \infty,
\end{equation}
The derivation of the small $x$ asymptotics of the form $\omega$ is based just  on the estimate (\ref{atzero}), i.e., no need for its 
extension, and it is much  easy to obtain,
$$
\omega=d\left(-\frac{\alpha^2}{8}\ln x-\frac{\alpha^2}{8}\right)-\frac{\alpha d\beta}{4}+O(x^{1-|\mathrm{\Im}(\alpha)|}))dx
$$
\begin{equation}\label{omegazero}
+ O(x^{2-|\mathrm{\Im}(\alpha)|}\ln x))dp+O(x^{2-|\mathrm{\Im}(\alpha)|}\ln x))dq, \quad x\to 0.
\end{equation}
As it has already been indicated, the derivations of formulae (\ref{omegainfty}) and (\ref{omegazero}) are straightforward.
However, because of the importance  of these formulae for our further analysis, we present the details of their derivations in the Appendix.

In  view of the  assumptions (\ref{techass}), estimates  (\ref{omegainfty}) and (\ref{omegazero}) yield the following
asymptotic representation for the form $w$,
\begin{equation}\label{omegaas1}
w = -d\left(\frac{\alpha^2}{8}\ln x +\frac{\alpha^2}{8}\right) + o(1), \quad x \to 0,
\end{equation}
and
\begin{equation}\label{omegaas2}
w = d\left(2\nu x + \nu^2\ln x +\nu^2\right) -\frac{i}{4}(b_+db_- -  b_- db_+) +\frac{x}{4}dx +\frac{\alpha d\beta}{4} + o(1), 
\quad x \to \infty.
\end{equation}
On the other hand, from (\ref{15}) and (\ref{16}) we have that
\begin{equation}\label{omegaas3}
w = -d\left(\frac{\alpha^2}{8}\ln x\right) + d\ln C_0 + o(1), \quad x \to 0,
\end{equation}
and
\begin{equation}\label{omegaas4}
w = d\left(2\nu x + \nu^2\ln x +\frac{x^2}{8}\right)  + d\ln C_{\infty} + o(1), 
\quad x \to \infty.
\end{equation}
The comparison of (\ref{omegaas1}) - (\ref{omegaas2}) and (\ref{omegaas3}) - (\ref{omegaas4})
implies that
$$
d\ln C_0 = -d\left(\frac{\alpha^2}{8}\right) 
$$
and
\begin{equation}
d\ln C_{\infty} =  d\nu^2 -\frac{i}{4}(b_+db_- -  b_- db_+) +\frac{\alpha d\beta}{4}.
\end{equation}
The last two equations mean that 
\begin{equation}\label{777}
d\ln\frac{C_{\infty}}{C_0} = d\left(\nu^2 +\frac{\alpha^2}{8} - i\nu\right) + \frac{\alpha d\beta}{4} - \frac{i}{2}b_+db_-
\end{equation}
(where we have also taken into account (\ref{nu})), or that
\begin{equation}\label{ccint}
\ln{\frac{C_\infty}{C_0}}=
\nu^2+\frac{\alpha^2}{8}-i\nu+\frac{1}{4}\intop (\alpha d\beta -2ib_+db_-) +c,
\end{equation}
where $c$ is the numerical constant, independent on $p$ and $q$.

Following \cite{ILT}, we introduce notation 
\begin{equation} \label{1}
e^{-4\pi i \rho}=\dfrac{\sin 2\pi (\sigma + \eta)}{\sin 2\pi\eta}.\nonumber
\end{equation} 
Using this and the connection formulae (\ref{connect}), (\ref{alphasigma}), we
can re-write the differential form \linebreak $\frac{1}{4}( \alpha d\beta-2ib_+db_-)$
as the differential form in variables $\eta$, $\rho$, $\sigma$ and $\nu$,
$$
\frac{1}{4}(\alpha d\beta -2ib_+db_-) = -8\pi i(\sigma d\eta +i\nu d\rho) +2\pi id\eta -(12 -48\sigma)\ln 2d\sigma
$$
\begin{equation}\label{etarho}
+(i\pi +4\ln2)\nu d\nu +(1-4\sigma)d\ln\frac{\Gamma(1-2\sigma)}{\Gamma(2\sigma)} +2i\nu d\ln\Gamma(1+i\nu).\nonumber
\end{equation}
Therefore, we can re-write (\ref{ccint}) as 
$$
\ln{\frac{C_\infty}{C_0}}=
\nu^2+4\sigma -8\sigma^2-i\nu +  2\pi i\eta- 12\sigma \ln 2  + 24\sigma^2\ln2 + \frac{i\pi\nu^2}{2}
+2\nu^2\ln2 
$$
\begin{equation}\label{ccint2}
-8\pi i\intop (\sigma d\eta +i\nu d\rho)  + \intop(1-4\sigma)d\ln\frac{\Gamma(1-2\sigma)}{\Gamma(2\sigma)}
+ \intop 2i\nu d\ln\Gamma(1+i\nu) +  c.
\end{equation}
It remains to evaluate the integrals in (\ref{ccint2}). For the integrals involving the $\Gamma$-functions
one gets,
\begin{equation}
\label{8}
\intop (1-4\sigma)d\ln\dfrac{\Gamma(1-2\sigma)}{\Gamma(2\sigma)}=\ln\dfrac{\Gamma(1-2\sigma)}{\Gamma(2\sigma)}-4\sigma+8\sigma^2+2\ln\Bigl(G(1-2\sigma)G(1+2\sigma)\Bigr)+c,
\end{equation}
\begin{equation}
\label{9}
\intop 2 i \nu d \ln \Gamma(1+i\nu)=i\nu-\nu^2-i\nu\ln(2\pi)+2\ln G(1+i\nu)+c,
\end{equation}
where $G(z)$ is the Barnes G-function and we have used the classical formula,
\begin{equation}\label{7}
\intop \ln\Gamma(x)dx=\dfrac{z(1-z)}{2}+\dfrac{z}{2}\ln(2\pi)+z\ln\Gamma(z)-\ln G(1+z)+c.\nonumber
\end{equation}
The most challenging, i.e., the first integral in (\ref{ccint2}) has already been evaluated in
\cite{ILT}. Here is the result.
\begin{equation}\label{2}
\intop \sigma d \eta +i \nu d\rho=\sigma\eta + i\nu\rho -\mathcal{W}(\sigma,\nu)+c,
\end{equation}
where the function $\mathcal{W}(\sigma,\nu)$ is expressed in terms of the dilogarithm $Li_2(z)$, 
\begin{equation}\label{3}
8\pi^2\mathcal{W}(\sigma,\nu)=Li_2(-e^{2\pi i(\sigma+\eta-i\frac{\nu}{2})})+Li_2(-e^{-2\pi i(\sigma+\eta+i\frac{\nu}{2})})-4\pi^2\eta^2+\pi^2\nu^2,
\end{equation}
Taking into account yet another classical formula,
\begin{equation}\label{4}
Li_2(e^{2\pi i z})=-2\pi i \ln \hat{G}(z)-2\pi i z \ln \dfrac{\sin(\pi z)}{\pi}-\pi^2 z(1-z)+\dfrac{\pi^2}{6}, \nonumber
\end{equation}
where
\begin{equation}\label{5}
\hat{G}(z)=\dfrac{G(1+z)}{G(1-z)},\nonumber
\end{equation}
and the elementary relation,
\begin{equation}\label{6}
2\cos\pi(\sigma+\eta\pm\frac{i\nu}{2})=e^{i\pi(\pm\sigma\mp\eta-\frac{i\nu}{2}-4\rho)},
\end{equation}
we arrive at the following final expression for the first integral in (\ref{ccint2})
$$
-8\pi i \intop \sigma d\eta + i\nu d\rho=-8\pi i\sigma \eta +2 \ln \dfrac{\hat{G}(\sigma+\eta+\frac{1-i\nu}{2})}{\hat{G}(\sigma+\eta+\frac{1+i\nu}{2})}-4\pi \eta^2 -i\pi \nu^2 +2i \ln(2\pi)\nu 
$$
\begin{equation}
\label{10}
-\dfrac{3\pi i \nu^2}{2}-2\pi \sigma^2 +6\pi i \eta^2+4\pi i \sigma\eta-2\pi i \sigma + 2\pi i \eta.
\end{equation}
Substituting  formulae \eqref{8}, \eqref{9}, and \eqref{10} in \eqref{ccint2} we arrive at the equation,
$$
\dfrac{C_\infty}{C_0}=c_1(2\pi)^{{i\nu}}2^{2\nu^2+{\sigma^2}{24}-12\sigma}e^{2\pi i(\eta^2-2\sigma\eta-\sigma^2+2\eta -\sigma) }
$$
\begin{equation}\label{final}
\times\dfrac{\Gamma(1-2\sigma)}{\Gamma(2\sigma)}\left(\dfrac{G(1+i\nu)G(1+2\sigma)G(1-2\sigma)\hat{G}(\sigma+\eta+\frac{1-i\nu}{2})}{\hat{G}(\sigma+\eta+\frac{1+i\nu}{2})}\right)^{{2}},
\end{equation}
where $c_1$ is a numerical constant. We know, that if $u=0, \sigma=\eta=\frac{1}{4}, \nu=0$, then $\tau=\mathrm{const} \cdot e^{\frac{x^2}{8}}$ and $C_\infty=C_0$. This choice of parameters satisfies conditions \eqref{set}.
Hence, 
$$
c_1=\dfrac{2^{\frac{3}{2}}e^{-i\frac{\pi}{4}}}{\pi (G(\frac{1}{2}))^4}.
$$
To complete the proof of Theorem \ref{the} we only need now to lift the technical assumption  (\ref{techass}). This
can be justified by noticing that the both sides of (\ref{final}) are analytic functions of the Riemann-Hilbert data.
(For the left hand side it follows from the general Birkhoff-Grothendieck-Malgrange theory.)
\begin{remark}
The variables $(\eta,\sigma)$ and $(-i\rho,\nu)$ are canonical variables. In fact, one has that \cite{ILT},
$$
\Omega = 32\pi i d\eta\wedge d\sigma = 32\pi d\rho\wedge d\nu.
$$
The function $\mathcal{W}$  was introduced in \cite{ILT} as the  {\it generating function} of the canonical transformation
$$
(\eta,\sigma)\to (-i\rho,\nu).
$$
Indeed, using  \eqref{3}, \eqref{6} and the fact that $Li'_2(z)=-z^{-1}\ln(z-1)$, one can show that \cite{ILT},
$$
\eta=\dfrac{\partial \mathcal{W}}{\partial\sigma},\quad\mbox{and}\quad i\rho=\dfrac{\partial \mathcal{W}}{\partial\nu}.
$$
The last equation is also equivalent to the integral formula (\ref{2}).
\end{remark}
In \cite{ILT}, and in fact earlier in the pioneering works \cite{Lis1}, \cite{Lis2},  the derivation of the constant terms in the asymptotics  of the tau-functions was based on the heuristic assumption (followed from the conformal block representation of the tau-functions) that
these constants are related to the generating functions of the relevant canonical transformations
between the canonical pairs associated with different critical points. In the case of equation (\ref{sineP3})
the points are $0$ and $\infty$ and the generating function is the function $\mathcal{W}$. This is a very
important conceptual point, and our analysis justifies it in the case of the Painlev\'e III equation (\ref{sineP3}).
It is also worth noticing that  this hamiltonian interpretation of the ratio $C_{\infty}/C_0$ is already present in formula (\ref{ccint}). 
Indeed, this formula tell us that the   logarithm of the ratio $C_{\infty}/C_0$ is,
up to the    elementary  function, $\nu^2 +\alpha^2/8 -i\nu$, the generating function of the canonical transformation
between the Cauchy data $(\alpha, \beta)$ and asymptotic at infinity data $(b_+, b_-)$.
\section{Proof of the ILT-Conjecture}
In \cite{ILT} different $\tau$-function was introduced
\begin{equation}
\label{12}
\tau_m(2^{-12}x^4)=(\tau(x))^{\frac{1}{2}}x^{\frac{1}{4}}e^{\frac{iu(x)}{4}},
\end{equation}
\begin{equation}
\label{13}
\tau_m(2^{-12}x^4)=\dfrac{2^{-12\sigma^2}x^{4\sigma^2}}{G(1+2\sigma)G(1-2\sigma)}(1+o(1)),\quad x\to 0,
\end{equation}
\begin{equation}
\label{14}
\tau_m(2^{-12}x^4)=\chi(\sigma,\nu)e^{\frac{i\pi\nu^2}{4}}2^{\nu^2}(2\pi)^{-\frac{i\nu}{2}}G(1+i\nu)x^{\frac{\nu^2}{2}+\frac{1}{4}}e^{\frac{x^2}{16}+\nu x}(1+o(1)), \quad x\to \infty.
\end{equation}
But from formulae \eqref{15}, \eqref{16}, \eqref{12} we also have
\begin{equation}
\label{17}
\tau_m(2^{-12}x^4)={C_0^{\frac{1}{2}}e^{-\frac{i\pi}{4}+i\pi\eta}2^{\frac{3}{2}-6\sigma}x^{4\sigma^2}}\dfrac{\Gamma(1-2\sigma)}{\Gamma(2\sigma)}(1+o(1)),\quad x\to 0,
\end{equation}
\begin{equation}
\label{18}
\tau_m(2^{-12}x^4)=C_\infty^{\frac{1}{2}}x^{\frac{\nu^2}{2}+\frac{1}{4}}e^{\frac{x^2}{16}+\nu x}(1+o(1)), \quad x\to \infty.
\end{equation}
From formulae \eqref{13}, \eqref{14}, \eqref{17}, \eqref{18} we get
\begin{equation}
\label{19}
\chi(\sigma,\nu,\eta)=\dfrac{C_\infty^{\frac{1}{2}}}{C_0^{\frac{1}{2}}}\dfrac{(2\pi)^{\frac{i\nu}{2}}2^{-\frac{3}{2}-\nu^2-12\sigma^2+6\sigma}e^{-\frac{i\pi\nu^2}{2}-i\pi\eta+\frac{i\pi}{4}}}{G(1+i\nu)G(1+2\sigma)G(1-2\sigma)}\left(\dfrac{\Gamma(2\sigma)}{\Gamma(1-2\sigma)}\right)^{\frac{1}{2}}.\nonumber
\end{equation}
Substituting here expression for $\frac{C_\infty}{C_0}$, we get the formula conjectured in \cite{ILT} 
$$
\chi(\sigma,\nu,\eta)=(2\pi)^{i\nu-\frac{1}{2}}e^{i\pi(\eta^2-2\sigma\eta-\sigma^2+\eta-\sigma-\frac{\nu^2}{4}+\frac{1}{8})}\dfrac{2^{-\frac{1}{4}}}{G^2(\frac{1}{2})}\dfrac{\hat{G}(\sigma+\eta+\frac{1-i\nu}{2})}{\hat{G}(\sigma+\eta+\frac{1+i\nu}{2})}.
$$
\section{Hamiltonian meaning of Malgrange-Bertola Differential form}
One can rewrite \eqref{omegaloc} 
$$
\omega=-\frac{xdx}{4}-\frac{\mathcal{H}dx}{4}-\frac{x\mathcal{H}_pdp}{4}-\frac{vu_pdp}{4}-\frac{x\mathcal{H}_qdq}{4}-\frac{vu_qdq}{4}
$$
$$
=-\frac{xdx}{4}-\frac{d(x\mathcal{H})}{4}+\frac{x\mathcal{H}_xdx}{4}-\frac{vdu}{4}+\frac{vu_xdx}{4}.
$$
It follows from \eqref{sineP3} that
$$
x\mathcal{H}_x+vu_x=\mathcal{H}.
$$
Using this formula we get
\begin{equation}\label{ham1}
\omega=-\frac{1}{4}d\left(\frac{x^2}{2}+x\mathcal{H}\right)+\frac{1}{4}\Bigl(\mathcal{H}dx-vdu\Bigr).
\end{equation}
We want to emphasize that all the objects are considered as the functions of the triple $(x, p, q)$ and
all the differentials are taken with respect to all  these three variables. 

From (\ref{ham1}) it follows that up to the multiplication by -4 and the subtraction of a total differential,
the Malgrange-Bertola form $\omega$ coincides with the form
\begin{equation}\label{ham2}
vdu-\mathcal{H}dx \equiv vu_xdx+vu_pdp+vu_qdq-\mathcal{H}dx.
\end{equation}
The restriction of this form on a trajectory of the Hamiltonian system (\ref{sineP3}), i.e. on the curve,
\begin{equation}\label{traj}
p =\mbox{const}, \quad q = \mbox{const},
\end{equation}
in the extended space of the monodromy data $\{(x, p, q)\}$ coincides with the form
$$
dS(x) =vu_xdx-\mathcal{H}dx,
$$
where $S(x)$ is the classical action evaluated on the trajectory (\ref{traj}). Hence, the  Malgrange-Bertola form $\omega$
can be treated as a natural extension of the canonical form $\mathcal{H}dx-vu_xdx$. It follows then, that the
tau-function itself can be identified with the  classical action. More precisely, along any classical trajectory,
we have that
\begin{equation}\label{tauS}
d\ln \tau \equiv \frac{d\ln\tau}{dx}dx = -\frac{1}{4}dS -\frac{1}{4}d(x \mathcal{H})\equiv 
\left(-\frac{1}{4}\frac{dS}{dx} - \frac{1}{4}\frac{d(x\mathcal{H})}{dx}\right)dx.
\end{equation}
Of course, this differential identity can be easily (after it is written) checked directly. In its turn,
it allows us to write the following representation for the ratio $C_{\infty}/C_{0}$ in terms of the regularized
action integral,
\begin{equation}\label{action1}
\ln\frac{C_{\infty}}{C_0}=\lim_{t_0\to 0}\lim_{t_1\to+\infty}\left(\intop_{t_0}^{t_1}\frac{\mathcal{H}-vu_x}{4}dx-\Bigl.x\mathcal{H}\Bigr|_{t_0}^{t_1}-\frac{t_1^2}{8}-2\nu t_1-\nu^2\ln t_1-\frac{\alpha^2}{8}\ln t_0\right).
\end{equation}
It is worth noticing that, unlike the integral $\int \mathcal{H}dx$, the action integral suits well to the  differentiation with respect 
to $p$ and $q$; indeed, after the relevant integration by part the remaining  integral term  would disappear 
in view of (\ref{sineP3}). Therefore, equation (\ref{action1}) provides us  with the possibility of an alternative derivation of 
our key formula (\ref{777}). This derivation would  be very similar to the evaluation of the
action integral of the McCoy-Tracy-Wu solution of the PIII equation in \cite{LZ}.
\begin{remark}
Observe that the extended (with respect to $x$, $p$, $q$)
differential of the form  $vdu - Hdx$ is the symplectic form $\Omega$.
Therefore, the fact that the Malgrange-Bertola form differs
from the form $-\frac{1}{4}(vdu-Hdx)$
by a total differential is a fact of general theory; indeed, the (extended) differentials
of the both forms coincide -- they both are the same 2-form, i.e.$-\frac{1}{4}\Omega$.
The additional information we are obtaining in \eqref{ham1} is the explicit evaluation of this total differential. This allows us to relate the tau-function and the action differential
explicitly which would be important for the alterantive evaluation
of the tau-constant via the action integral.
\end{remark}
It might seem quite surprising  that one needed to start with the  Malgrange-Bertola form
in order to discover a rather  simple differential identity (\ref{tauS}). The absence of the very idea
that the logarithm of the tau-function might differ from the classical action just by a  total
differential partially explains this. We now believe that the similar fact should be
true for any isomonodrtomy tau-function, although it has been apparently missing  
in the general monodrmy theory of linear systems {\footnote{In 2000, the first co-author together
with Percy Deift tried to use technique of \cite{LZ} for evolution of the constant factors
in the asymptotics of the Painlev\'e V tau-function associated with the sine-kernel.
We failed then because we did not have  the analog of the relation (\ref{tauS})
for the Painlev\'ve V tau - function we were working with. Perhaps,
it would make sense to revisit the issue now (although, the relevant constant
factors have already been evaluated since then).}}.


\section*{Acknowledgements}
The work was supported in
part by NSF grant DMS-1361856, and by the  SPbGU grant N 11.38.215.2014. The authors are also grateful to M. Bertola for very useful comments.
\section{Appendix. The Derivation of Estimates  (\ref{omegainfty}) and (\ref{omegazero}).}
Substitution of the extended asymptotics (\ref{inftyext}) into the right hand side of (\ref{omegaloc}) leads 
to the following expressions for its individual terms.
\begin{itemize}
\item $
 { \dfrac{xu_x^2}{8}}=-\dfrac{b_+^2e^{2ix}x^{2i\nu}}{8}-\dfrac{b_-^2e^{-2ix}x^{-2i\nu}}{8}-\dfrac{b_+^2e^{2ix}x^{2i\nu-1}}{32}(6i\nu^2+2\nu+3i)$
\begin{multline*}
+\dfrac{b_-^2e^{-2ix}x^{-2i\nu-1}}{32}(6i\nu^2-2\nu+3i)-\nu-\dfrac{\nu^2}{x}+\dfrac{b_+^4e^{4ix}x^{4i\nu-1}}{64}+\dfrac{b_-^4e^{-4ix}x^{-4i\nu-1}}{64}+O(x^{-2+6|{\Im \nu}|});
\end{multline*}
\item
$
-{ \dfrac{x}{4}(\cos u -1)}=\dfrac{b_+^2e^{2ix}x^{2i\nu}}{8}+\dfrac{b_-^2e^{-2ix}x^{-2i\nu}}{8}-\nu-\dfrac{b_+^4e^{4ix}x^{4i\nu-1}}{64}-\dfrac{b_-^4e^{-4ix}x^{-4i\nu-1}}{64}
$
\begin{multline*}
+\dfrac{b_+^2e^{2ix}x^{2i\nu-1}}{32}(6i\nu^2+2\nu-i)-\dfrac{b_-^2e^{-2ix}x^{-2i\nu-1}}{32}(6i\nu^2-2\nu-i)+O(x^{-2+6|{\Im \nu}|});
\end{multline*}
\item
$
{ \dfrac{x^2}{4}u_p \sin u}=\dfrac{b_+b_{+p}e^{2ix}x^{2i\nu}}{16}(6i\nu^2+5\nu-i)+\dfrac{b_-b_{+p}x}{4}+\dfrac{b_+b_{+p}e^{2ix}x^{2i\nu+1}}{4}-\dfrac{b_+^3b_{+p}e^{4ix}x^{4i\nu}}{16}
$
\begin{multline*}
-\dfrac{3b_-^3b_{+p}e^{-2ix}x^{-2i\nu}}{64}+\dfrac{b_-b_{+p}\nu}{4} +\dfrac{ib_+^2\nu_{p}e^{2ix}x^{2i\nu+1}\ln x}{4}-\dfrac{ib_-^2\nu_{p}e^{-2ix}x^{-2i\nu+1}\ln x}{4}
\\
-\dfrac{ib_+^4\nu_{p}e^{4ix}x^{4i\nu}\ln x}{16}-\dfrac{b_-^2\nu_{p}e^{-2ix}x^{-2i\nu}\ln x}{16}(6\nu^2+2i\nu-1)-\dfrac{b_+^2\nu_{p}e^{2ix}x^{2i\nu}\ln x}{16}(6\nu^2-2i\nu-1)
\\
+\dfrac{b_{-p}b_{+}x}{4}+\dfrac{b_-b_{-p}e^{-2ix}x^{-2i\nu+1}}{4}-\dfrac{3b_+^3b_{-p}e^{2ix}x^{2i\nu}}{64}-\dfrac{b_-^3b_{-p}e^{-4ix}x^{-4i\nu}}{16}+\dfrac{b_+b_{-p}\nu}{4}
\\
+\dfrac{ib_-^4\nu_{p}e^{-4ix}x^{-4i\nu}\ln x}{16}-\dfrac{b_-b_{-p}e^{-2ix}x^{-2i\nu}}{16}(6i\nu^2-5\nu-i)+\dfrac{b_+^2\nu_{p}e^{2ix}x^{2i\nu}}{8}(3i\nu-1)
\\
-\dfrac{b_-^2\nu_{p}e^{-2ix}x^{-2i\nu}}{8}(3i\nu+1)+\nu_p\nu+O(x^{-1+6|{\Im \nu}|});
\end{multline*}
\item
$
{ \dfrac{x^2}{4}u_q \sin u} = \Bigl\{ p\to q\Bigr\};
$
\item
$
{ \dfrac{x^2}{4}u_x u_{px}}=-\dfrac{b_+b_{+p}e^{2ix}x^{2i\nu+1}}{4}-\dfrac{ib_+^2\nu_{p}e^{2ix}x^{2i\nu+1}\ln x}{4}-\dfrac{b_-b_{-p}e^{-2ix}x^{-2i\nu+1}}{4}
$
\begin{multline*}
+\dfrac{ib_-^2\nu_{p}e^{-2ix}x^{-2i\nu+1}\ln x}{4}-\dfrac{b_+^2\nu_{p}e^{2ix}x^{2i\nu}}{8}(3i\nu+1)-\dfrac{b_+b_{+p}e^{2ix}x^{2i\nu}}{16}(6i\nu^2+\nu+3i)+\dfrac{b_-b_{+p}x}{4}
\\
+\dfrac{b_+b_{-p}x}{4}-\nu_p\nu+\dfrac{b_+b_{-p}\nu}{4}+\dfrac{b_-b_{+p}\nu}{4}+\dfrac{b_+^2\nu_{p}e^{2ix}x^{2i\nu}\ln x}{16}(6\nu^2-2i\nu+3)+\dfrac{b_+^3b_{+p}e^{4ix}x^{4i\nu}}{16}
\\
+\dfrac{ib_+^4\nu_{p}e^{4ix}x^{4i\nu}\ln x}{16}+\dfrac{b_-^2\nu_{p}e^{-2ix}x^{-2i\nu}\ln x}{16}(6\nu^2+2i\nu+3)+\dfrac{b_-b_{-p}e^{-2ix}x^{-2i\nu}}{16}(6i\nu^2-\nu+3i)
\\
+\dfrac{b_-^2\nu_{p}e^{-2ix}x^{-2i\nu}}{8}(3i\nu-1)+\dfrac{b_-^3b_{-p}e^{-4ix}x^{-4i\nu}}{16}-\dfrac{ib_-^4\nu_{p}e^{-4ix}x^{-4i\nu}\ln x}{16}
\\
-\dfrac{b_-^3b_{+p}e^{-2ix}x^{-2i\nu}}{64}-\dfrac{b_+^3b_{-p}e^{2ix}x^{2i\nu}}{64}+O(x^{-1+6|{\Im \nu}|});
\end{multline*}
\item
$
{ \dfrac{x^2}{4}u_x u_{qx}} = \Bigl\{ p\to q\Bigr\}; 
$
\item
$
 { \dfrac{xu_xu_p}{4}}=\dfrac{ib_+b_{+p}e^{2ix}x^{2i\nu}}{4}-\dfrac{b_+^2\nu_{p}e^{2ix}x^{2i\nu}\ln x}{4}+\dfrac{ib_+b_{-p}}{4}-2\nu_p\nu\ln x-\dfrac{ib_-b_{+p}}{4}
$
\begin{multline*}
-\dfrac{ib_-b_{-p}e^{-2ix}x^{-2i\nu}}{4}-\dfrac{b_-^2\nu_{p}e^{-2ix}x^{-2i\nu}\ln x}{4}+O(x^{-1+4|{\Im \nu}|});
\end{multline*}
\item
$
{ \dfrac{xu_xu_q}{4}}= \Bigl\{ p\to q\Bigr\}.
$
\end{itemize}
It is quite remarkable that the substitution of these long expressions into the right hand side of (\ref{omegaloc}) yields 
to the compact formula (\ref{omegainfty}).

For asymptotics at zero we get the following estimates.
$$
\dfrac{xu_x^2}{8}-\dfrac{x}{4}(\cos u -1)=\dfrac{\alpha^2}{8x}+O(x^{1-|\Im \alpha|}),
$$
$$
\dfrac{x^2 u_p \sin u}{4}=O(x^{2-|\mathrm{\Im}(\alpha)|}),\quad \dfrac{x^2 u_q \sin u}{4}=O(x^{2-|\mathrm{\Im}(\alpha)|}),
$$
$$
\dfrac{x^2u_x u_{px}}{x}=\dfrac{\alpha\alpha_p}{4}+O(x^{2-|\mathrm{\Im}(\alpha)|}),\quad
\dfrac{x^2u_x u_{qx}}{x}=\dfrac{\alpha\alpha_q}{4}+O(x^{2-|\mathrm{\Im}(\alpha)|}),
$$
$$
\dfrac{xu_x u_p}{4}=\dfrac{\alpha \alpha_p \ln x}{4} + \dfrac{\alpha \beta_p}{4}+O(x^{2-|\mathrm{\Im}(\alpha)|}\ln x),
$$
$$
\dfrac{xu_x u_q}{4}=\dfrac{\alpha \alpha_q \ln x}{4} + \dfrac{\alpha \beta_q}{4}+O(x^{2-|\mathrm{\Im}(\alpha)|}\ln x).
$$
These equations yield at once (\ref{omegazero}).

\end{document}